\documentclass[onecolumn,11pt,aps,nofootinbib,superscriptaddress,tightenlines,showkeys]{revtex4}

\usepackage[T1]{fontenc}
\usepackage[utf8]{inputenc}

\usepackage[cmex10]{amsmath}
\usepackage{graphics}
\usepackage{dsfont}
\usepackage{amssymb}
\usepackage{graphicx}
\usepackage{mathrsfs}
\usepackage{units}
\usepackage{pgfplots}
\usepackage{enumitem}
\usepackage[colorlinks=true,linkcolor=black,citecolor=black,plainpages=false,pdfpagelabels]{hyperref}
\hypersetup{pdftitle={Eigenvalue estimates for the resolvent of a non-normal matrix}}

\usepackage{amsthm}
\newtheorem{theorem}{Theorem}[section]
\newtheorem{lemma}[theorem]{Lemma}
\newtheorem{proposition}[theorem]{Proposition}
\newtheorem{corollary}[theorem]{Corollary}
\newtheorem{definition}{Definition}[section]

\newenvironment{example}[1][Example]{\begin{trivlist}
\item[\hskip \labelsep {\bfseries #1}]}{\end{trivlist}}
\newenvironment{remark}[1][Remark]{\begin{trivlist}
\item[\hskip \labelsep {\bfseries #1}]}{\end{trivlist}}

\newcommand{\tr}{\textnormal{tr}}

\newcommand{\R}{\mathbb{R}}

\newcommand{\id}{\ensuremath{\mathds{1}}}

\begin{document}
\title{Explicit Frames for Deterministic Phase Retrieval via PhaseLift}
\author{Michael Kech}
\email{kech@ma.tum.de}
\affiliation{Department of Mathematics, Technische Universit\"{a}t M\"{u}nchen, 85748 Garching, Germany}

\date{\today}

\begin{abstract}
We explicitly give a frame of cardinality $5n-6$ such that every signal in $\mathbb{C}^n$ can be recovered up to a phase from its associated intensity measurements via the PhaseLift approach. Furthermore, we give explicit linear measurements with $4r(n-r)+n-2r$ outcomes that enable the recovery of every positive semidefinite $n\times n$ matrix of rank at most $r$.
\end{abstract}

\maketitle

\tableofcontents

\section{Introduction and Main Result}
Phase Retrieval is the task of reconstructing a signal $x\in\mathbb{C}^n$ up to a phase from intensity measurements.

In \cite{balan2006signal} it was shown that $m\ge 4n-2$ generic intensity measurements suffice to discriminate any two signals in $\mathbb{C}^n$ up to a phase. With a similar approach this result was slightly improved to $m\ge 4n-4$ in \cite{conca2014algebraic} \footnote{In the context of pure state tomography, \cite{kech2,mondragon2013determination,carmeli2015many} show that the $4n-4$ bound also holds for von Neumann measurements. In addition similar bounds for the recovery of low-rank matrices with constrained measurements are provided in \cite{kech2}.}.  The bound  $m\ge 4n-4$ is known to be close to optimal. More precisely, by relating phase retrieval to the problem of embedding complex projective space in Euclidean space, it was shown in \cite{heinosaari2013quantum} that, up to terms at most logarithmic in $n$, $m\ge 4n-4$ intensity measurements are necessary to discriminate any two signals in $\mathbb{C}^n$ up to a phase. However, \cite{balan2006signal,conca2014algebraic} do not provide a tractable recovery scheme.
A result indicating that some redundancy is needed in order to allow for computationally efficient phase retrieval is given in \cite{fickus2014phase}.

There have been several approaches that do provide recovery schemes \cite{balan2009painless,alexeev2014phase,bandeira2014phase}, in the present paper however we focus on the approach of \cite{candes2015phase} known as PhaseLift. Their approach consists of two steps: First, phase retrieval is lifted to the problem of recovering rank one Hermitian matrices from linear measurements. Secondly, by means of a convex relaxation, the recovery problem is formulated as a trace norm minimization over a spectrahedron. The authors of \cite{candes2015phase} then prove that $\mathcal{O}(n)$ intensity measurements suffice to recover a signal modulo phase with high probability by solving the relaxed optimization problem. Furthermore, stability guarantees for the recovery were established in \cite{candes2013phaselift,candes2014solving}. While these convex relaxations are in principal tractable, solving them becomes computationally expensive with increasing signal dimension \cite{candes2015wirtinger}.

However, \cite{candes2015phase,candes2013phaselift,candes2014solving} still leave room for improvement. For example, by working with Gaussian random vectors additional structure that might facilitate the use of PhaseLift is not incorporated and also from a practical point of view Gaussian random vectors might not be desirable. Recently, it was shown that a partial derandomization of PhaseLift can be achieved by using spherical designs \cite{gross2015partial,kueng2015spherical}. The purpose of the present paper is similar. However, rather than drawing the measurements from a smaller, possibly better structured set, we aim for finding explicit measurements that allow for phase retrieval via PhaseLift. Another deterministic approach to the phase retrieval problem was introduced in \cite{bodmann2015stable}. They improved their results in \cite{bodmann2016algorithms}, providing recovery algorithms together with explicit error bounds for phase retrieval with $6n-3$ frame vectors. 

Our contribution is the following: We explicitly give $5n-6$ intensity measurements from which every signal in $\mathbb{C}^n$ can be reconstructed up to a phase using PhaseLift. More precisely, for $k\in\{1,\dots,2n-3\}$ let
\begin{align}\label{eq7}
v_k:=\begin{pmatrix}
1,&
x_k\ e^{\frac{i\pi}{2n}},&
x_k^2\ e^{2\frac{i\pi}{2n}},&
\dots,&
x_k^{n-1}\ e^{(n-1)\frac{i\pi}{2n}}
\end{pmatrix}^t,\ x_k\in\mathbb{R}\setminus\{0\}.
\end{align}
Furthermore denote by $\{e_i\}_{i\in\{0,\dots,n-1\}}$ the standard orthonormal basis of $\mathbb{C}^n$.
\begin{theorem}\label{thm0}
If $x_1<x_2<\ldots<x_{2n-3}$, then every signal $x\in\mathbb{C}^n$ can be reconstructed up to a phase from the $5n-6$ intensities 
\begin{align*}
\{|\langle e_0,x\rangle|^2,\dots,|\langle e_{n-1},x\rangle|^2,|\langle v_1,x\rangle|^2,|\langle\overline{v}_1,x\rangle|^2,\dots,|\langle v_{2n-3},x\rangle|^2,|\langle\overline{v}_{2n-3},x\rangle|^2\}
\end{align*}
via PhaseLift.
\end{theorem}
This result is stated more carefully in Section \ref{1} as Corollary \ref{cor1}. Its proof relies on the results of \cite{chen2013uniqueness}.

Let us highlight three features of this result: 
\begin{enumerate}
\item To our knowledge the $5n-6$ is the smallest number of intensity measurements that allow for a uniform and computationally tractable recovery.
\item Results based on random intensity measurements typically guarantee that the recovery succeeds with high probability if the number of measurements exceeds a given threshold which is usually determined up to a multiplicative constant. As opposed to this, Theorem \ref{thm0} comes with two advantages that might be desirable from a practical point of view: First, the recovery is not just guaranteed to succeed with high probability but indeed works deterministically. Secondly, since the measurements are given explicitly there is no need for finding a suitable value for the threshold. 
\item Theorem \ref{thm0} merely requires $5n-6$ intensity measurements. This illustrates that $n$ additional measurements as compared to the nearly optimal bound of \cite{balan2006signal} suffice to render PhaseLift feasible.
\end{enumerate}

The approach we take originates from low-rank matrix recovery \cite{candes2009exact,candes2010matrix,candes2010power,recht2010guaranteed,gross2010quantum} and indeed the previous results can be generalised to this setting: In Section \ref{1}, we give an explicit family of linear measurements with $4r(n-r)+n-2r$ outcomes from which every positive $n\times n$ matrix of rank at most $r$ can be recovered by means of a semidefinite program. This strongly relies on the construction of the null spaces of such measurements given in \cite{chen2013uniqueness}. Our contribution is to explicitly characterize the orthogonal complements of these null spaces leading to the proofs of our main results. 

Finally we also prove a weak stability result in Section \ref{stab}, showing that the reconstruction error is linear in the error scale. As we do not know how to estimate the constant of proportionality appearing in the stability bound, this result is not of practical relevance, but might give a roadmap for proving stability in the future. However, we provide some numerical results that might indicate the constant's qualitative behaviour.

\section{Preliminaries}\label{0}
Let us first fix some notation. By $M(n,q)\ (M(n,q,\mathbb{R}))$ we denote the set of complex (real) $n\times q$ matrices. The transpose (conjugate transpose) of a matrix $A\in M(n,q)$ is denoted by $A^t$ ($A^*$). For $i\in\{0,\dots,n-1\},\ j\in\{0,\dots,q-1\}$, we denote the entry in the $i$-th row and $j$-th column of a matrix $A\in M(n,q)$ by $A_{ij}$ \footnote{Note that the indices we use to label matrices begin with $0$, not with $1$.}. By $H(n)$ we denote the real vector space of Hermitian $n\times n$ matrices. We equip $H(n)$ with the Hilbert-Schmidt inner product and $\|\cdot\|_2$ denotes the Frobenius norm. By $\mathcal{S}^n$ we denote the set of positive semidefinite $n\times n$ matrices and by $\mathcal{S}_r^n\subseteq \mathcal{S}^n$ we denote the subset of positive semidefinite  matrices of rank at most $r$. In the following we assume that $r\in\{1,\dots,\lceil n/2\rceil-1\}$ \footnote{$\lceil k\rceil:=$ the smallest integer $i$ such that $i\geq k$.}. The set of linear maps $M:H(n)\to \mathbb{R}^m$ is denoted by $\mathcal{M}(m)$. 

\begin{definition}($m$-measurement.)
An $m$-measurement is an element of $\mathcal{M}(m)$.
\end{definition}
In the following we denote an $m$-measurement simply by measurement if we do not want to specify $m$.
\begin{remark}
For each $m$-measurement $M$ there exists a unique $G:=(G_{1},\dots,G_{m})\in H(n)^m$ such that
\begin{align*}
M(X)=\big( \text{tr}(G_1X),\dots,\text{tr}(G_mX) \big)
\end{align*}
for all $X\in H(n)$. By $M_G$ we denote the $m$-measurement associated in this way to an $G\in H(n)^m$. In the following we sometimes use this identification to speak of elements $G\in H(n)^m$ as $m$-measurements.
\end{remark}
\begin{definition}($r$-complete.)
A measurement $M$ is called $r$-complete iff $M(X)\neq M(X^\prime)$ for all $X\in\mathcal{S}_r^n$ and $X^\prime\in\mathcal{S}^n$ with $X\neq X^\prime$. A tuple $G\in H(n)^m$ is called $r$-complete iff $M_G$ is $r$-complete.
\end{definition}
Given a measurement $M$ and a measurement outcome $b=M(X),\ X\in\mathcal{S}_r^n$, consider the following well-known semi-definite program \cite{candes2009exact,candes2010power,recht2010guaranteed} \footnote{This is a convex relaxation of the rank minimization problem.}
\begin{gather}\label{sdp}
\begin{split}
\text{minimize}\ \text{tr}(Y)\ \ \ \ \ \ \ \ \ \ \\
\text{subject to}\ Y\geq 0,\ M(Y)=b.
\end{split}
\end{gather}
The significance of the $r$-complete property is due to the following observation: 
\begin{proposition}\label{prop1}
Let $M$ be an $r$-complete measurement and let $X\in\mathcal{S}_r^n$. If $b=M(X)$, then $X$ is the unique minimizer of the semidefinite program \eqref{sdp}.
\end{proposition}
\begin{proof}
Let $X\in\mathcal{S}_r(\mathbb{C}^n)$ be a Hermitian matrix of rank at most $r$ and let $M$ be an $r$-complete measurement. Then, $X$ is the unique feasible point of the spectrahedron 
\begin{align}\label{sph}
\{Y\in H(n): Y\geq 0,\ M(Y)=M(X)\}.
\end{align} This follows immediately from $\{Y\in H(n): Y\geq 0,\ M(Y)=M(X)\}=\{Y\in \mathcal{S}^n:M(Y)=M(X)\}$ and the definition of $r$-complete.
\end{proof}
\begin{remark}
Note that if $\id\in\text{Range}(M^*)$, the $r$-complete property also is necessary for a deterministic reconstruction via the semidefinite program \eqref{sdp}.
\end{remark}
This shows that for an $r$-complete measurement the semidefinite program \eqref{sdp} reduces to a feasibility problem.

Finally, let us state the observation of \cite{chen2013uniqueness,CarmeliTeikoJussi1} which gives a useful characterization of the $r$-complete property:
\begin{proposition}\label{proprank}
A measurement $M$ is $r$-complete if and only if every nonzero $X\in\text{Ker}(M)$ has at least $r+1$ positive eigenvalues.
\end{proposition}
\begin{proof}
Consider the set $\Delta:=\{Y-Z:Y\in\mathcal{S}_r^n,\ Z\in\mathcal{S}^n\}$ and note that every $X\in\Delta$ has at most $r$ positive eigenvalues. Furthermore, note that a measurement $M$ is $r$-complete if and only if $\Delta\cap \text{Ker}(M)\setminus\{0\}=\emptyset$. 

Now assume that every $X\in\text{Ker}(M)\setminus\{0\}$ has at least $r+1$ positive eigenvalues. Since every $Y\in\Delta$ has at most $r$ positive eigenvalues we find $Y\notin\text{Ker}(M)\setminus\{0\}$, i.e. $\Delta\cap \text{Ker}(M)\setminus\{0\}=\emptyset$.

Conversely, assume that $M$ is $r$-complete. $\Delta$ clearly contains all matrices with at most $r$ positive eigenvalues and hence $\text{Ker}(M)\setminus\{0\}$ cannot contain an element with $r$ or less positive eigenvalues.

\end{proof}
\begin{remark}
If every nonzero $X\in\text{Ker}(M)$ has at least $r+1$ positive eigenvalues, then every nonzero $X\in\text{Ker}(M)$ also has at least $r+1$ negative eigenvalues since $X\in\text{Ker}(M)$ implies $-X\in\text{Ker}(M)$.
\end{remark}

\section{Reconstruction of Low-Rank Positive Matrices}\label{1}
Our approach relies on \cite{chen2013uniqueness} where a method to construct the null spaces of $r$-complete $m$-measurements for $m=4r(n-r)+n-2r$ is provided. Their construction is based on the ideas of \cite{cubitt2008dimension}, details can be found in Appendix A of \cite{chen2013uniqueness}.

First, we focus on the phase retrieval problem.
\begin{theorem}\label{thm1}
Let
\begin{align*}
G:=\left(e_0e_0^*,\dots,e_{n-1}e_{n-1}^*,\frac{v_1v_1^*}{\|v_1v_1^*\|_2},\frac{\overline{v}_1\overline{v}_1^*}{\|\overline{v}_1\overline{v}_1^*\|_2},\dots,\frac{v_{2n-3}v_{2n-3}^*}{\|v_{2n-3}v_{2n-3}^*\|_2},\frac{\overline{v}_{2n-3}\overline{v}_{2n-3}^*}{\|\overline{v}_{2n-3}\overline{v}_{2n-3}^*\|_2}\right),
\end{align*}
where the $v_i$ are defined in Equation \eqref{eq7}. If $x_1<x_2<\ldots<x_{2n-3}$, then the measurement $M_G$ is $1$-complete.
\end{theorem}
The proof of this theorem can be found in Section \ref{proofs}.
\begin{remark}
From the proof of this result it is easily seen that the kernel of $M_G$ is independent of the choice of the $x_i$. Thus, for the purpose of robustness, the $x_i$ should be chosen such that the smallest singular value of $M_G$ is maximized.
\end{remark}
Let us next state Theorem \ref{thm0} in a more precise way.
\begin{corollary}(Phase Retrieval via PhaseLift.)\label{cor1}
Let $M$ be a measurement given by Theorem \ref{thm1} and let $x\in\mathbb{C}^n$. If $b=M(xx^*)$, then $xx^*$ is the unique minimizer of the semidefinite program \eqref{sdp}.
\end{corollary}
By Proposition \ref{prop1}, this is an immediate consequence of Theorem \ref{thm1}. 

Let us next focus on the recovery of low-rank positive matrices. This, however, requires some further definitions: First, let
\begin{align}\label{cn}
C^n_r:=\{X\in H(n):\text{tr}(Xe_ie_j^*)=0,\ 2r-1\le i+j\le 2(n-r)-1,\ i\neq j,\}.
\end{align}
E.g. $C^n_1\subseteq H(n)$ is the subspace of $n\times n$ diagonal matrices and $C^7_3$ is the subspace of $H(7)$ of the from
\begin{gather*}
\begin{pmatrix}
*&*&*&*&*&0&0\\
*&*&*&*&0&0&0\\
*&*&*&0&0&0&*\\
*&*&0&*&0&*&*\\
*&0&0&0&*&*&*\\
0&0&0&*&*&*&*\\
0&0&*&*&*&*&*
\end{pmatrix}.
\end{gather*}
For $x\in\mathbb{R}\setminus\{0\}$, $k\in\{2r-1,\dots,2(n-r)-1\}$, define the Hermitian matrices $R_k(x),I_k(x)\in (C_r^n)^\bot$ by \footnote{As $R_k(x),I_k(x)\in (C_r^n)^\bot$ both have vanishing diagonal and since they are hermitian, it suffices to define all elements above the diagonal.}
\begin{align*}
(R_k(x))_{jl}:=\delta_{j+l,k}x^j,\ j,l\in\{0,\dots,n-1\},\ j>l,\\
(I_k(x))_{jl}:=i\delta_{j+l,k}x^j,\ j,l\in\{0,\dots,n-1\},\ j>l,
\end{align*}
where $\delta_{i,j}$ denotes the Kronecker delta.
E.g. for $n=5$, $r=2$ these are 
\begin{tiny}
\begin{gather*}
R_3(x)=\begin{pmatrix}
0&0&0&1&0\\
0&0&x&0&0\\
0&x&0&0&0\\
1&0&0&0&0\\
0&0&0&0&0\\
\end{pmatrix},
I_3(x)=\begin{pmatrix}
0&0&0&i&0\\
0&0&ix&0&0\\
0&-ix&0&0&0\\
-i&0&0&0&0\\
0&0&0&0&0\\
\end{pmatrix},
R_4(x)=\begin{pmatrix}
0&0&0&0&1\\
0&0&0&x&0\\
0&0&0&0&0\\
0&x&0&0&0\\
1&0&0&0&0\\
\end{pmatrix},
I_4(x)=\begin{pmatrix}
0&0&0&0&i\\
0&0&0&ix&0\\
0&0&0&0&0\\
0&-ix&0&0&0\\
-i&0&0&0&0\\
\end{pmatrix},\\
R_5(x)=\begin{pmatrix}
0&0&0&0&0\\
0&0&0&0&1\\
0&0&0&x&0\\
0&0&x&0&0\\
0&1&0&0&0\\
\end{pmatrix},
I_5(x)=\begin{pmatrix}
0&0&0&0&0\\
0&0&0&0&i\\
0&0&0&ix&0\\
0&0&-ix&0&0\\
0&-i&0&0&0\\
\end{pmatrix}.
\end{gather*}
\end{tiny}
\begin{theorem}\label{thm2}
Let $G_0$ be a basis of $C_r^n$ and let $x_1,x_2,\dots,x_{r}\in\mathbb{R}\setminus\{0\}$ with  $x_1<x_2<\ldots<x_r$. For $k\in\{2r-1,\dots,2(n-r)-1\}$ define
\begin{align*}
G_k:=(I_k(x_1),R_k(x_1),\dots,I_k(x_r),R_k(x_r)).
\end{align*}
and let $G:=G_0\cup G_{2r-1}\cup \dots\cup G_{2(n-r)-1}$ \footnote{For tuples of Hermitian matrices $X:=(X_1,\dots,X_i)\in H(n)^i$, $Y:=(Y_1,\dots,Y_j)\in H(n)^j$ we define their union $X\cup Y$ to be the tuple $X\cup Y:=(X_1,\dots,X_i,Y_1,\dots,Y_j)\in H(n)^{i+j}$.}. Then the measurement $M_{G}$ is $r$-complete  and $|G|=4r(n-r)+n-2r$.
\end{theorem}
\begin{remark}
If an $m$-measurement is injective when restricted to $\mathcal{S}_r^n$, it was shown in \cite{heinosaari2013quantum,kech1} that, up to terms at most logarithmic in $n$, we have $m\ge 4r(n-r)$. Furthermore, in \cite{heinosaari2013quantum,kech2} it was shown that there indeed exist injective $m$-measurements for $m=4r(n-r)$. Thus, it might be worth noting that the measurements given by Theorem \ref{thm2} solely require $n-2r$ additional measurement outcomes as compared to the nearly optimal bound $4r(n-r)$.
\end{remark}
Finally, by Proposition \ref{prop1}, the measurements given by Theorem \ref{thm2} allow for the recovery of low-rank positive matrices.
\begin{corollary}(Recovery of low-rank positive matrices.)
Let $M$ be a measurement given by Theorem \ref{thm2} and let $X\in\mathcal{S}^n_r$.  If $b=M(X)$, then $X$ is the unique minimizer of the semidefinite program \eqref{sdp}.
\end{corollary}
\section{Stability}\label{stab}
In this section we discuss the stability of $r$-complete measurements.

Assume there is an error term $E\in H(n)$ that perturbs the matrix $X_r\in\mathcal{S}_r^n$ we intend to recover to the matrix $X=X_r+E$. Measuring with an $r$-complete measurement $M$ yields the perturbed outcome $b=M(X)$. Clearly, the matrix $X_r$ cannot always be perfectly recovered from the outcome $b$, however, if $\|M(E)\|_2$ is small, there is a recovery procedure that yields a matrix close to $X_r$. For that purpose, consider the following well-known optimization problem
\begin{gather}\label{sdp2}
\begin{split}
\text{minimize}\ \text{tr}(Y)\ \ \ \ \ \ \ \ \ \ \ \ \\
\text{subject to}\ Y\geq 0,\ \|M(Y)-b\|_2\leq\epsilon
\end{split}
\end{gather}
where $\epsilon\geq 0$ is a constant representing the error scale.
\begin{theorem}(Stable recovery of low-rank positive matrices.)\label{thmstab}
Let $M$ be an $r$-complete measurement and let $\epsilon>0$. There is a constant $C_M>0$ independent of $\epsilon$ such that for all $X_r\in\mathcal{S}_r^n$ and $E\in H(n)$ with $\|M(E)\|_2\leq\epsilon$, any minimizer $Y$ of \eqref{sdp2} for $b=M(X_r+E)$ satisfies
\begin{align*}
\|Y-X_r\|_2\leq C_M\epsilon.
\end{align*}
\end{theorem}
\begin{remark}
In the proof of this theorem we show that $C_M\leq \frac{2}{\sigma_{\min}}(1+\frac{1}{\kappa})$ where $\sigma_{min}$ is the smallest singular value of $M$ and $\kappa:=-\max_{Z\in\text{Ker}(M),\|Z\|_2=1}\lambda_{n-r}(Z)$ \footnote{$\lambda_{n-r}$ is defined later this section.}. However we do not know how to compute $\kappa$ for a given $r$-complete measurement $M$ and hence we cannot make this bound more explicit.
\end{remark}
The proof of this theorem can be found in Section \ref{proofs}. In order for this result to be a practical stability guarantee, one would have to estimate the constant $C_M$. At this point we do not know how this can be achieved. In order to indicate the magnitude of the constant $C_M$, let us next present some numerical results. For this purpose consider the tuple
\begin{align*}
G_n:=\left(e_0e_0^*,\dots,e_{n-1}e_{n-1}^*,\frac{I_1(1)}{\|I_1(1)\|_2},\frac{R_1(1)}{\|R_1(1)\|_2},\dots,\frac{I_{2n-3}(1)}{\|I_{2n-3}(1)\|_2},\frac{R_{2n-3}(1)}{\|R_{2n-3}(1)\|_2}\right)
\end{align*}
and note that by Theorem \ref{thm2} the associated measurement $M_{G_n}$ is $1$-complete. Figure \ref{fig1}  presents numerical results that might indicate the scaling of $C_{M_{G_n}}$ for the sequence of measurements $(M_{G_n})_{n\in\mathbb{N}}$.
\begin{figure}[ht]
\includegraphics[trim=0cm 8cm 1cm 9cm,clip, width=12cm]{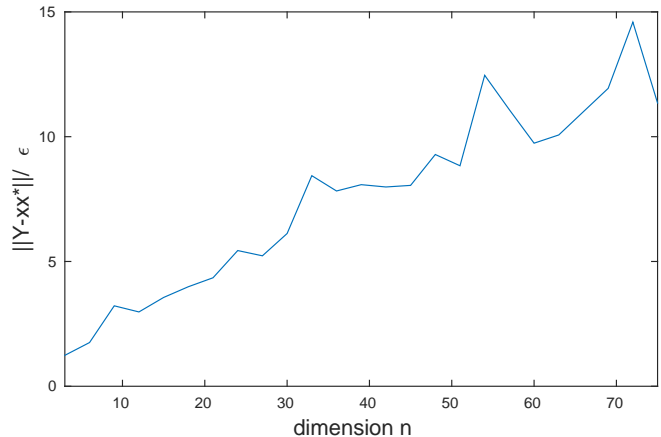}
\centering
\caption{\label{fig1} For each $n\in\{3,6,\dots,75\}$ we choose uniformly at random a normalized vector $x\in\mathbb{C}^{n}$ and an error term $f\in\mathbb{R}^{5n-6}$ with $\|f\|_2\leq\epsilon:=10^{-3}$. Then we run the program \eqref{sdp2} with the outcome $b=M_{G_n}(xx^*)+f$. The figure shows the maximum value of $\|Y-xx^*\|_2/\epsilon$ for $2200$ repetitions where $Y$ is the minimizer of \eqref{sdp2}.}
\end{figure}

Just like in \cite{candes2013phaselift}, this recovery scheme can also be used for the phase retrieval problem. For a Hermitian matrix $A\in H(n)$, we denote by $\text{Eig}(A)\in\R^n$ the tuple of eigenvalues of $A$ ordered decreasingly together with their multiplicities. Furthermore, we define $\lambda_i(A):=\text{Eig}(A)_{i-1},\ i\in\{1,\dots,n\}$.

\begin{proposition}(Stability for Phase Retrieval.)\label{propstab}
Let $X=xx^*+E$, where $x\in\mathbb{C}^n$ is the signal and $E\in H(n)$ is an error term. Let $M$ be a $1$-complete measurement and let $\epsilon\geq \|M(E)\|_2$. Furthermore, let $Y$ be any minimizer of the optimization problem \eqref{sdp2} for $b=M(X)$ and set $\hat{x}:=\sqrt{\lambda_1(Y)}x^\prime$ where $x^\prime\in S^{n-1}$ is an eigenvector of $Y$ with eigenvalue $\lambda_1(Y)$. Then
\begin{align*}
\|xx^*-\hat{x}\hat{x}^*\|_2\leq 2C_M\epsilon,
\end{align*}
where $C_M$ is the constant given by Theorem \ref{thmstab}.
Furthermore, for some $\varphi\in[0,2\pi)$ we have
\begin{align*}
\|x-e^{i\varphi}\hat{x}\|_2\leq \frac{2\sqrt{2}C_M}{\|x\|_2}\epsilon.
\end{align*}
\end{proposition}
This result follows from Theorem \ref{thmstab} by a straightforward computation. The proof is given in Section \ref{proofs}.
\begin{remark}
The proofs of \ref{lemstab} shows that the above stability results also hold true the following recovery scheme:
\begin{gather}
\begin{split}
\text{minimize}\ \|M(Y)-b\|_2\\
\text{subject to}\ Y\geq 0,\ \ \ \ 
\end{split}
\end{gather}
where $M$ is $r$-complete and $b=M(X_r+E)$, $X_r\in\mathcal{S}_r^n$.
\end{remark}

\section{Technical Appendix}\label{proofs}
Let us first introduce some notation we use throughout this section.  Let $A\in M(n,q)$, $i\in\{0,\dots,n-1\},\ j\in\{0,\dots,q-1\}$. By $A_{,i}$ we denote the $(n-1)\times q$ matrix obtained from $A$ by deleting the $i$-th row and by $A^{,j}$ we denote the $n\times (q-1)$ matrix obtained from $A$ by deleting the $j$-th column. By $A\{i\}$ we denote the $i$-th row of $A$ and by $A[j]$ we denote the $j$-th column of $A$. Furthermore, for $k\in\{0,\dots,n+q-2\}$, we denote the $k$-th anti-diagonal of $A$ by $A(k)$, i.e. $A(k):=(A_{ij})_{i+j=k}$ \footnote{The ordering is such that the matrix element with smaller $i$ comes first.}. 
\subsection{Proof of Theorem \ref{thm2}}
Since Theorem \ref{thm1} is obtained by manipulating the measurements obtained from Theorem \ref{thm2} we begin by proving the latter. The construction we give in the following yields a more general class of $r$-complete measurements than the ones given by Theorem \ref{thm2} and it strongly relies on the notion of totally non-singular matrices.
\begin{definition}(Totally non-singular.)
A matrix $A\in M(n,q)$ is called totally non-singular if $A$ has no vanishing minor.
\end{definition}
The following lemma is a central ingredient for the construction given in the following.
\begin{lemma}\label{lem1}
Let $q\in\{1,\dots,n-1\}$ and let $A\in M(n,q)$ be totally non-singular. Then, there exists a totally non-singular matrix $B\in M(n,n-q)$ such that $A^*B=0$.
\end{lemma}
\begin{proof}
We give a proof by induction in the dimension $n$ for $q$ fixed. 

\textbf{Base case.} Let us begin with the base case $n=q+1$. Note that for a given $A\in M(q+1,q)$ there always exists a nonzero matrix (actually just a vector) $B\in M(q+1,1)$ such that $A^*B=0$, in particular if $A$ is totally non-singular.

Since $B$ exists, it is enough to prove that if $A$ is totally non-singular $B$ is totally non-singular as well: Assume for a contradiction that $B$ is not totally non-singular, i.e. that $B$ has a vanishing entry. By permuting rows we can assume $A$ and $B$ to be of the form 
\begin{align*}
A=\begin{pmatrix} F \\ D \end{pmatrix},\ \ B=\begin{pmatrix} 0 \\ E \end{pmatrix}
\end{align*}
for some matrices $F\in M(1,q)$, $D\in M(q,q)$ and $E\in M(q,1)$. But then
\begin{align*}
A^*B=F^*0+D^*E=D^*E=0.
\end{align*} 
In particular this implies that the $q\times q$ submatrix $D$ of $A$ is singular, contradicting the fact that $A$ is totally non-singular by assumption.

\textbf{Induction step.} Assume the claim holds for an $n>q$ and let $A\in M(n+1,q)$ be totally non-singular. Note that for each $i\in\{0,\dots,n\}$, the $n\times q$ matrix $A_{,i}$ is totally non-singular since $A$ is totally non-singular. Thus, by the induction hypothesis, we can find for each $i\in\{0,\dots,n\}$ a totally non-singular matrix $C_i\in M(n,n-q)$ such that $A_{,i}^*C_i=0$. For $i\in\{0,\dots,n\},j\in\{0,\dots,n-q\}$ let $C(i,j)\in M(n+1,n+1-q)$ be the matrix with $C(i,j)_{,i}^{,j}=C_i$ and $0$ else. Then, for all $i\in\{0,\dots,n\},j\in\{0,\dots,n-q\}$,  $C(i,j)_{,i}^{,j}$ is totally non-singular, $C(i,j)[j]=0$ and $A^*C(i,j)=0$ by construction.

\textbf{Step 1.} First, for each $i\in\{0,\dots,n\}$, we deform $C(i,0)$ into a matrix $\tilde{C}(i,0)\in M(n+1,n+1-q)$ with the following properties:
\begin{itemize}
\item[1.] $A^*\tilde{C}(i,0)=0$,
\item[2.] $\tilde{C}(i,0)_{,i}^{,0}$ is totally non-singular,
\item[3.] All $(n+1-q)\times (n+1-q)$ minors of $\tilde{C}(i,0)_{,i}$ are nonzero.
\end{itemize}  

Let $i\in\{0,\dots,n\}$. For $\sigma:=(k_0,\dots,k_{n-q})\in \Sigma:=\{(l_0,\dots,l_{n-q}):0\leq l_0<\ldots<l_{n-q}\leq n-1\}$ define the projection $P_\sigma:M(n+1,n+1-q)\to M(n+1-q,n+1-q)$ by $P_\sigma(X)\{j\}:=(X_{,i})\{k_j\}$ for all $X\in M(n+1,n+1-q),j\in\{0,\dots,n-q\}$. Now let $\sigma\in\Sigma$, and set $E_\sigma:=P_\sigma(C(i,0))$. By permuting rows we can assume $A$ and $C(i,0)$ to be of the form
\begin{align}\label{decomp}
A=\begin{pmatrix} F \\ D \end{pmatrix},\ \ C(i,0)=\begin{pmatrix} E_\sigma \\ F_\sigma \end{pmatrix}
\end{align}
for some matrices $F\in M(n+1-q,q)$, $D\in M(q,q)$ and $F_\sigma\in M(q,n+1-q)$.

Next, we show that there is a vector $u_\sigma=\begin{pmatrix} v_\sigma \\ w_\sigma \end{pmatrix}$ \footnote{The direct sum decomposition of $u_\sigma$ is with respect to the decomosition given by Equation \eqref{decomp}, i.e. $A^*u_\sigma=F^*v_\sigma+D^*w_\sigma$.}, where $v_\sigma\in\mathbb{C}^{n+1-q}$, $w_\sigma\in\mathbb{C}^{q}$, such that $A^*u_\sigma=0$ and $\det(E_\sigma+P_\sigma(u_\sigma e_{0}^*))=\det(E_\sigma+v_\sigma e_{0}^*)\neq 0$: Since $C_i$ is totally non-singular, $E_\sigma$ has rank $n-q$. Thus we can find a vector $v_\sigma\in\mathbb{C}^{n+1-q}$ such that $\det(E_\sigma+v_\sigma e_{0}^*)\neq 0$ \footnote{Note that $E_\sigma[0]=0$ by construction of $C(i,0)$.}. Finally, we just have to ensure that $A^*u_\sigma =0$. Since $D$ is totally non-singular there is a vector $w_\sigma\in\mathbb{C}^{q}$ such that $D^*w_\sigma=-F^*v_\sigma$ and this gives $A^*\begin{pmatrix} v_\sigma \\ w_\sigma \end{pmatrix}=F^*v_{\sigma}+D^*w_\sigma=0$. Repeating this construction, we can find a collection of vectors $\{u_\sigma\}_{\sigma\in\Sigma}\subseteq\mathbb{C}^{n+1}$ such that for all $\sigma\in\Sigma$ we have $A^*u_\sigma=0$ and $\det\big(P_\sigma(C(i,0)+u_\sigma e_{0}^*)\big)\neq 0$.

Next, for distinct $\sigma_1,\sigma_2\in\Sigma$, define the mapping $K(\lambda):=C(i,0)+u_{\sigma_1}e_{0}^*+\lambda u_{\sigma_2}e_{0}^*$, $\lambda\in\mathbb{C}$ and note that by construction $A^*K(\lambda)=0$ for all $\lambda\in\mathbb{C}$. Note that by construction $K(\lambda)_{,i}^{,0}=C_i$ is totally non-singular for all $\lambda$. Furthermore, the $(n+1-q)\times (n+1-q)$ minors $\text{det}\big(P_{\sigma_1}(K(\lambda))\big)$ and $\text{det}\big(P_{\sigma_2}(K(\lambda))\big)$ can be considered as polynomials in $\lambda$. The polynomial equations $\text{det}\big(P_{\sigma_1}(K(\lambda))\big)=0$ and $\text{det}\big(P_{\sigma_2}(K(\lambda))\big)=0$ are non-trivial: For $\lambda=0$ we have $\text{det}\big(P_{\sigma_1}(K(0))\big)=\text{det}\big(P_{\sigma_1}(C(i,0)+u_{\sigma_1} e_{0}^*)\big)\neq 0$ by construction. For $\lambda$ large one can consider $\frac1{\lambda}u_{\sigma_1}e_{0}^*$ as a small perturbation to $u_{\sigma_2}e_{0}^*$. Thus, using linearity of the determinant in the $0$-th column,  we conclude that
\begin{align*}
\text{det}\big(P_{\sigma_2}(K(\lambda))\big)&=\text{det}\big(P_{\sigma_2}(C(i,0)+u_{\sigma_1}e_{0}^*+\lambda u_{\sigma_2}e_{0}^*)\big)\\
&=\lambda\cdot\text{det}\big(P_{\sigma_2}(C(i,0)+u_{\sigma_2}e_{0}^*)+\frac1{\lambda} P_{\sigma_2}(u_{\sigma_1}e_{0}^*)\big)\neq 0
\end{align*}
for large enough $\lambda$ by the continuity of the determinant and the fact that $\text{det}\big(P_{\sigma_2}(C(i,0)+u_{\sigma_2}e_{0}^*)\big)\neq 0$ by construction. A non-trivial polynomial equation in one variable just has a finite set of solutions and hence the set 
\begin{gather*}
\{\lambda\in\mathbb{C}:\text{det}\big(P_{\sigma_1}(K(\lambda))\big)=0\vee \text{det}\big(P_{\sigma_2}(K(\lambda))\big)=0\}\\
=\{\lambda\in\mathbb{C}:\text{det}\big(P_{\sigma_1}(K(\lambda))\big)=0\}\cup \{\lambda\in\mathbb{C}:\text{det}\big(P_{\sigma_2}(K(\lambda))\big)=0\}
\end{gather*}
is finite. In particular there is an $a_{\sigma_2}\in\mathbb{C}$ such that $\text{det}\big(P_{\sigma_1}(K(a_{\sigma_2}))\big)\neq 0$ and $\text{det}\big(P_{\sigma_2}(K(a_{\sigma_2}))\big)\neq 0$ \footnote{In fact this holds for almost all $a_{\sigma_2}\in\mathbb{C}$.}. Applying the same argument to $L(\lambda):=C(i,0)+u_{\sigma_1}e_{0}^*+a_{\sigma_2} u_{\sigma_2}e_{0}^*+\lambda u_{\sigma_3}e_{0}^*$, $\lambda\in\mathbb{C}$, where $\sigma_3\in\Sigma$ is distinct from $\sigma_1,\sigma_2$, yields an $a_{\sigma_3}\in\mathbb{C}$ such that $\text{det}\big(P_{\sigma_1}(L(a_{\sigma_3}))\big)\neq 0$, $\text{det}\big(P_{\sigma_2}(L(a_{\sigma_3}))\big)\neq 0$ and $\text{det}\big(P_{\sigma_3}(L(a_{\sigma_3}))\big)\neq 0$ \footnote{Also in this case we obtain a finite set of non-trivial polynomial equations in $\lambda$ and thus the argument given before can be applied to find $a_{\sigma_3}$.}. Finally, since $|\Sigma|$ is finite, we can inductively apply the argument to obtain a matrix $\tilde{C}(i,0)=C(i,0)+u_{\sigma_1}e_{0}^*+\sum_{\sigma\in\Sigma,\sigma\neq\sigma_1}a_\sigma u_{\sigma}e_{0}^*$ with the desired properties.

\textbf{Step 2.} 
Secondly, we construct for each $i\in\{0,\dots,n\}$ a matrix $D_i\in M(n+1,q)$ with the following properties:
\begin{itemize}
\item[1.] $A^*D_i=0$.
\item[2.] $(D_i)_{,i}$ is totally non-singular.
\end{itemize}  
Let $i\in\{0,\dots,n\}$. Let $D_i(\lambda_1,\dots,\lambda_{n-q}):=\tilde{C}(i,0)+\sum_{j=1}^{n-q}\lambda_{j}C(i,j)$ where $\lambda_{j}\in\mathbb{C}$, $j\in\{1,\dots,n-q\}$, and note that by construction we have $A^*D_i(\lambda_1,\dots,\lambda_{n-q})=0$ for all $\lambda_1,\dots,\lambda_{n-q}\in\mathbb{C}$. By choosing  $(\lambda_{1},\hdots,\lambda_{n-q})$ appropriately one can make sure that $(D_i(\lambda_1,\dots,\lambda_{n-q}))_{,i}$ is totally non-singular: First let $G(\lambda):=\tilde{C}(i,0)+\lambda C(i,1),\ \lambda\in\mathbb{C}$. Just like in Step 1, the minors of $G(\lambda)_{,i}^{,0}$ and $G(\lambda)_{,i}^{,1}$ together with the $(n+1-q)\times(n+1-q)$ minors of $G(\lambda)_{,i}$ yield a finite set of polynomial equations in $\lambda$. All of these polynomial equations are non-trivial: For $\lambda=0$ none of the minors of $G(0)_{,i}^{,0}=C_i$ and none of the $(n+1-q)\times(n+1-q)$ minors of $G(0)_{,i}=\tilde{C}(i,0)_{,i}$ vanish by construction of $\tilde{C}(i,0)$. For large $\lambda$ one can consider $\frac1{\lambda}\tilde{C}(i,0)$ as a small perturbation to $C(i,1)$. Hence, for large enough $\lambda$, none of the minors of $\frac1{\lambda}G(\lambda)_{,i}^{,1}$ vanishes by the fact that $C(i,1)_{,i}^{,1}=C_i$ is totally non-singular by construction and the continuity of the minors. Thus, just like in Step 1, we conclude that there are just finitely many values of $\lambda$ for which any of these polynomials vanishes. In particular there is an $\lambda_1\in\mathbb{C}$ such that both $G(\lambda_1)_{,i}^{,0}$ and $G(\lambda_1)_{,i}^{,1}$ are totally non-singular and all $(n+1-q)\times(n+1-q)$ minors of $G(\lambda_1)_{,i}$ are nonzero. Applying the same argument to $H(\lambda):=\tilde{C}(i,0)+\lambda_1C(i,1)+\lambda C(i,2),\ \lambda\in\mathbb{C},$ yields an $\lambda_2\in\mathbb{C}$ such that $H(\lambda_2)_{,i}^{,0}$, $H(\lambda_2)_{,i}^{,1}$ and $H(\lambda_2)_{,i}^{,2}$ are totally non-singular and all $(n+1-q)\times(n+1-q)$ minors of $H(\lambda_2)_{,i}$ are nonzero. Choosing the values for $\lambda_j$, $j\in\{1,\dots,n-q\}$, inductively in this fashion finally yields a matrix $D_i$ with the desired properties.

\textbf{Step 3.} To complete the induction step we choose by a similar argument as in Step 1 and Step 2 before parameters $\gamma_j\in\mathbb{C},\ j\in\{1,\dots,n\},$ in $B:=D_0+\sum_{j=1}^{n}\gamma_j D_j$ such that $B_{,i}$ is totally non-singular for each $i\in\{0,\dots,n\}$, i.e. such that $B$ is totally non-singular: First define $I(\lambda):=D_0+\lambda D_1$, $\lambda\in\mathbb{C}$. Clearly $I(0)_{,0}=(D_0)_{,0}$ is totally non-singular by construction of $D_0$. Furthermore, for large $\lambda$ , $\frac1{\lambda}D_0$ can be considered as a small perturbation to $D_1$. Thus, for $\lambda$ large enough, $\frac1{\lambda}I(\lambda)_{,1}$ is totally non-singular by construction of $D_1$ and the continuity of the minors. Hence, all the minors of $I(\lambda)_{,0}$ and $I(\lambda)_{,1}$ yield non-trivial polynomial equations in $\lambda$ and therefore there are just finitely many values for $\lambda$ for which any of these minors vanishes. In particular there is a $\gamma_1\in\mathbb{C}$ such that both $I(\gamma_1)_{,0}$ and $I(\gamma_1)_{,1}$ are totally non-singular. Applying the same argument to $J(\lambda):=D_0+\gamma_1D_2+\lambda D_2$ yields a $\gamma_2\in\mathbb{C}$ such that $J(\gamma_2)_{,0}$, $J(\gamma_2)_{,1}$ and $J(\gamma_2)_{,2}$ are totally non-singular. Continuing to choose the $\gamma_i$, $i\in\{1,\dots,n\}$, inductively in this fashion then yields a totally non-singular matrix $B$ with $A^*B=0$.
\end{proof}
\begin{lemma}\label{cor2}
Let $q\in\{1,\dots,n-1\}$ and let $A\in M(n,q,\mathbb{R})$ be totally non-singular. Then, there exists a totally non-singular matrix $B\in M(n,n-q,\mathbb{R})$ such that $A^tB=0$.
\end{lemma}
\begin{proof}
The arguments given in the proof of Lemma \ref{lem1} also apply to real numbers.
\end{proof}
For $k\in\{1,\dots,2n-3\}$, define the inclusion in the $k$-th antidiagonal $\iota_k:\mathbb{C}^{\gamma(n,k)}\to H(n)$ by
\begin{align*}
(\iota_k(v))_{jl}:=\frac{1}{\sqrt{2}}\left\{
		\begin{array}{ll}
			v_j & \mbox{if } j+l=k, j<l \\
			v_l^* & \mbox{if } j+l=k, l<j \\
			0 & \mbox{else}
		\end{array}
	\right.
\end{align*}
where 
\begin{align*}
\gamma(n,k)=\left\{
		\begin{array}{ll}
			\lceil k/2\rceil& \mbox{if } k\le n-1 \\
			\lceil n-1-k/2\rceil & \mbox{if } k>n-1
		\end{array}
	\right.
\end{align*} 
is the length of the upper half of the $k$-th antidiagonal. By expanding in the generalised Gell-Mann orthonormal basis of $H(n)$, it is easily seen that the inclusion of real vectors in the same antidiagonal preserves the standard inner product, i.e. for $k\in\{1,\dots,2n-3\}$ we have
\begin{align}\label{eq1}
\tr\big(\iota_k(v)\iota_k(w)\big)=\langle v,w\rangle,\ \forall v,w\in\mathbb{R}^{\gamma(n,k)}.
\end{align}
Furthermore, the inclusion of an imaginary and a real vector in the same antidiagonal yields Hilbert-Schmidt orthogonal matrices, i.e. for  $k\in\{1,\dots,2n-3\}$ we have
\begin{align}\label{eq2}
\tr\big(\iota_k(v)\iota_k(iw)\big)=0,\ \forall v,w\in\mathbb{R}^{\gamma(n,k)},
\end{align}
and finally that inclusions of vectors in different antidiagonals also yield Hilbert-Schmidt orthogonal matrices, i.e. for $k,j\in\{1,\dots,2n-3\}$  with $k\neq j$ we have
\begin{align}\label{eq3}
\tr\big(\iota_k(v)\iota_j(w)\big)=0,\ \forall v\in\mathbb{C}^{\gamma(n,k)},w\in\mathbb{C}^{\gamma(n,j)}.
\end{align}

The following theorem is the main result of the present paper.
\begin{theorem}\label{thm3}
Let $G_0$ be a basis of $C_r^n$ \footnote{$C_r^n$ was defined in Equation \eqref{cn}.}. Furthermore, for $k\in\{2r-1,\dots,2(n-r)-1\}$, let $A_k,A^\prime_k\in M(\gamma(n,k),r,\mathbb{R})$ be totally non-singular and define the tuple
\begin{align*}
G_k:=\big(\iota_k(A_k[0]),\iota_k(iA^\prime_k[0]),\iota_k(A_k[1]),\iota_k(iA^\prime_k[1]),\dots,\iota_k(A_k[r-1]),\iota_k(iA^\prime_k[r-1])\big).
\end{align*}
Then $G:=G_0\cup G_{2r-1}\cup G_{2r}\cup \dots \cup G_{2(n-r)-1}$ is $r$-complete and $|G|=4r(n-r)+n-2r$.
\end{theorem}
\begin{proof}
The idea of the proof is to use  Lemma \ref{cor2} to determine a basis of the null space of $M_G$ such that the construction of \cite{chen2013uniqueness} can be applied. We do this in the first step of the proof. In the second step of the proof we use the construction of \cite{chen2013uniqueness} to show that $M_G$ indeed is $r$-complete.

\textbf{Step 1.} First, by Lemma \ref{cor2}, there are totally non-singular $B_k,B_k^\prime\in M(\gamma(n,k),\gamma(n,k)-r,\mathbb{R}),\ k\in\{2r+1,\dots,2(n-r)-3\}$,  such that 
\begin{align}\label{eq4}
\begin{split}
A_k^tB_k&=0,\\
(A_k^\prime)^tB^\prime_k&=0.
\end{split}
\end{align}
Now let
\begin{gather*}
G^\bot_k:=\big(\iota_k(B_k[0]),\iota_k(iB^\prime_k[0]),\dots,\iota_k(B_k[\gamma(n,k)-r-1]),\iota_k(iB^\prime_k[\gamma(n,k)-r-1])\big),\\
\ k\in\{2r+1,\dots,2(n-r)-3\}
\end{gather*}
and let $G^\bot_k=(0)$ for $k\in\{2r-1,2r,2(n-r)-2,2(n-r)-1\}$. In the remainder of this first step we prove that $G^\bot_{2r-1}\cup \dots\cup G^\bot_{2(n-r)-1}$ is a basis of $\text{Ker}(M_G)$: For $k\in\{1,\dots,2n-3\}$, let $Q_k:=\{X\in H(n): X(j)=0\ \forall j\neq k\ \land\ X_{ii}=0\text{ for } 2i=k\}$ be the subspace of Hermitian matrices with vanishing diagonal and non-vanishing entries only in the $k$-th antidiagonal. By Equation \eqref{eq3}, $H(n)$ can be decomposed into the following mutually orthogonal subspaces:
\begin{align}\label{eq5}
H(n)=C_r^n\oplus Q_{2r-1}\oplus\dots\oplus Q_{2(n-r)-1}.
\end{align}
Note that $\text{Span}(G_k\cup G_k^\bot)\subseteq Q_k$ for all $k\in\{2r-1,\dots,2(n-r)-1\}$. Hence, by the decomposition \eqref{eq5}, to show that $G^\bot_{2r-1}\cup \dots\cup G^\bot_{2(n-r)-1}$ is a basis of $\text{Ker}(M_G)$ it suffices to prove that for $k\in\{2r-1,\dots,2(n-r)-1\}$ the matrices $G^\bot_k\cup G_k$ span the subspace $Q_k$ and that $\text{Span}(G_k^\bot)\subseteq \text{Ker}(M_G)$. First observe that indeed $\text{Span}(G_k^\bot)\subseteq \text{Ker}(M_G)$ for every $k\in\{2r+1,\dots,2(n-r)-3\}$: Note that for every $k\in\{2r+1,\dots,2(n-r)-3\}$,
\begin{gather}
\begin{split}
\text{tr}\big( \iota_k(A_k[l])\iota_k(B_k[j])\big)=\langle A_k[l],B_k[j]\rangle=(A_k^tB_k)_{lj}=0,\ \ \\ 
\text{tr}\big( \iota_k(iA^\prime_k[l])\iota_k(iB^\prime_k[j])\big)=\langle A^\prime_k[l],B^\prime_k[j]\rangle=((A^\prime_k)^tB^\prime_k)_{lj}=0,\\ 
\forall l\in\{0,\dots,r-1\},\ j\in\{0,\dots,\gamma(n,k)-r-1\},\ \ \ \ \
\end{split}
\end{gather}
by equations \eqref{eq1} and \eqref{eq4}. Furthermore,
\begin{gather}\label{eq8}
\begin{split}
\text{tr}\big( \iota_k(iA^\prime_k[l])\iota_k(B_k[j])\big)=0,\ \ \ \ \ \ \ \ \ \ \ \ \ \ \\ 
\text{tr}\big( \iota_k(A_k[l])\iota_k(iB^\prime_k[j])\big)=0,\ \ \ \ \ \ \ \ \ \ \ \ \ \ \\ 
\forall l\in\{0,\dots,r-1\},\ j\in\{0,\dots,\gamma(n,k)-r-1\},
\end{split}
\end{gather}
by Equation \eqref{eq2}. I.e. $\text{Span}(G_k^\bot)$ is orthogonal to $\text{Span}(G_k)$ and thus $\text{Span}(G_k^\bot)\subseteq\text{Ker}(M_G)$.

To conclude the first step, we prove that $G^\bot_k\cup G_k$ spans the subspace $Q_k$ for every $k\in\{2r-1,\dots,2(n-r)-1\}$: Let $k\in\{2r-1,\dots,2(n-r)-1\}$. Since $A_k$ is totally non-singular, the columns of $A_k$ are linearly independent and the same argument applies to $A_k^\prime$. Hence, by the equations \eqref{eq1} and \eqref{eq2}, $G_k$ is a tuple of linearly independent Hermitian matrices. The same argument applies to $G^\bot_k$, $k\in\{2r+1,\dots,2(n-r)-3\}$. But we have seen that $\text{Span}(G_k)$ is orthogonal to $\text{Span}(G^\bot_k)$ for $k\in\{2r+1,\dots,2(n-r)-3\}$. Furthermore, for $k\in\{2r-1,\dots,2(n-r)-1\}$,  $|G^\bot_k|+|G_k|=2(\gamma(n,k)-r)+2r=2\gamma(n,k)=\dim Q_k$ and thus $G^\bot_k\cup G_k$ indeed spans $Q_k$. 

Finally, observe that 
\begin{align*}
|G|&=\dim C_r^n+\sum_{i=2r-1}^{2(n-r)-1}|G_i|=\sum_{i=1}^{2r-2}2\gamma(n,i)+n+\sum_{i=1}^{2(n-2r)+1}2r\\
&=(2r)^2-2(2r)+n+2r(2(n-2r)+1)\\
&=4r(n-r)+n-2r.
\end{align*}

\textbf{Step 2.} 
In the second step, we essentially reproduce the construction of \cite{chen2013uniqueness} and some ideas of \cite{cubitt2008dimension}. We  show in the following that every nonzero matrix $X\in\text{Ker}(M_G)$ has at least $r+1$ positive and $r+1$ negative eigenvalues and this concludes the proof by Proposition \ref{proprank}.

Let $X\in\text{Ker}(M_G)$ be arbitrary. By the interlaced eigenvalue Theorem (Theorem 4.3.15 of \cite{horn2012matrix}) it suffices to prove that there is an $2(r+1)\times 2(r+1)$ principal submatrix of $X$ with $r+1$ positive and $r+1$ negative eigenvalues. We conclude the proof by finding such a submatrix: There is a smallest number $k\in\{2r+1,\dots,2(n-r)-3\}$ such that $X$ has non-vanishing entries in the $k$-th antidiagonal. First note that either the real or the imaginary part of the $k$-th antidiagonal does not vanish. Let us consider the case where the real part does not vanish, the other case can be shown analogously. The real part of the $k$-th antidiagonal of $\text{Ker}(M_G)$ is spanned by the  $\gamma(n,k)-r$ real matrices of $G_k^\bot$, i.e. each  $X\in\text{Ker}(M_G)$ is a linear combination of the $\gamma(n,k)-r$ real matrices of $G_k^\bot$. But then there have to be at least $2(r+1)$ non-vanishing entries in the $k$-th antidiagonal of $X$ because otherwise there would be a vanishing $(\gamma(n,k)-r)\times (\gamma(n,k)-r)$ minor of $B_k$ and this contradicts the fact that $B_k$ is totally non-singular (For more details see Lemma 9 of \cite{cubitt2008dimension}.). I.e. there is a $2(r+1)\times 2(r+1)$ principal submatrix of $X$ of the from:
\begin{gather}\label{eq6}
\begin{split}
\begin{pmatrix}
0&0&0&\dots&0&0&\overline{x}_1\\
0&0&0&\dots&0&\overline{x}_2&\overline{y}^1_{1}\\
0&0&0&\dots&\overline{x}_3&\overline{y}^2_1&\overline{y}^1_{2}\\
\vdots &\vdots&\vdots& &\vdots&\vdots&\vdots\\
0&0&x_3&\dots&0&\overline{y}^2_{2r-2}&\overline{y}^1_{{2r-1}}\\
0&x_2&{y}^2_{1}&\dots&{y}^2_{2r-2}&0&\overline{y}^1_{2r}\\
x_1&{y}^1_{1}&{y}^1_{2}&\dots&{y}^1_{2r-1}&{y}^1_{2r}&0
\end{pmatrix},
x_i\in\mathbb{C}\setminus\{0\},\ i\in\{1,\dots,r+1\},
\end{split}
\end{gather}
where $y_i^j\in\mathbb{C},\ j\in\{1,\hdots,r\},i\in\{1,\hdots,2(r+1)-2j\},$ are arbitrary.

Finally, we show by induction that a matrix of this form has at least $r+1$ positive and $r+1$ negative eigenvalues: The claim clearly holds for $r=0$. Now assume the claim holds for $r\in\mathbb{N}_0$. Let $Y$ be a $2(r+2)\times 2(r+2)$ matrix that is of the form illustrated in Equation \eqref{eq6}. Then, one can obtain a principal  $2(r+1)\times 2(r+1)$ submatrix $Y^\prime$ of $Y$ that is of the same form by e.g. deleting the first and last row as well as the first and last column of $Y$. Thus, by the induction hypothesis and the interlaced eigenvalue Theorem (Theorem 4.3.15 of \cite{horn2012matrix}), $Y$ has at least $r+1$ positive and $r+1$ negative eigenvalues. A straightforward calculation shows that $\det(Y)\cdot\det(Y^\prime)<0$. Since the determinant of a matrix is the product of its eigenvalues, the claim follows from $\det(Y)\cdot\det(Y^\prime)<0$.
\end{proof}

In the following the $r=1$ case is of particular interest because Theorem \ref{thm1} is obtained from this case by choosing the totally non-singular matrices appropriately.
\begin{corollary}\label{cor3}
Let $G_{0}:=(e_0e_0^*,\dots,e_{n-1}e_{n-1}^*)$. Furthermore let $w_k,v_k\in \mathbb{R}^{\gamma(n,k)},\ k\in\{1,\dots,2n-3\}$, be such that every entry of $v_k$ and every entry of $w_k$ is nonzero.
Then $G:=G_0\cup \big(\iota_1(v_1),\iota_1(iw_1)\big)\cup \dots \cup \big(\iota_{2n-3}(v_{2n-3}),\iota_{2n-3}(iw_{2n-3})\big)$ is $1$-complete and $|G|=5n-6$.
\end{corollary}
\begin{proof}
First, note that $G_{0}$ is a basis of $C_1^n$. Furthermore as by assumption all entries of matrices $w_k,v_k\in \mathbb{R}^{\gamma(n,k)}\simeq M(\gamma(n,k),1,\mathbb{R}),\ k\in\{1,\dots,2n-3\},$ are nonzero, we conclude that all their minors are nonzero\footnote{The minors of a $m\times 1$ matrix are simply the entries of the matrix.}. Consequently the matrices $w_k,v_k\in M(\gamma(n,k),1,\mathbb{R}),\ k\in\{1,\dots,2n-3\},$ are totally non-singular. Hence $G_{0}$ and $G_k:=\big(\iota_1(v_1),\iota_1(iw_1)),\ k\in\{1, \hdots, 2n-3\}$ fulfil the conditions of Theorem \ref{thm3} for $r=1$ and thus $G=G_0\cup G_1\cup\hdots\cup G_{2n-3}$ is $1$-complete.
\end{proof}
\begin{example}\label{ex}
For $i\in\{1,\dots,2n-3\}$, we can choose $w_{i}=v_{i}=\sqrt{2}e$, where $e:=(1,\dots,1)\in\mathbb{R}^{\gamma(n,i)}$ is the vector with a one in every component. Altogether this yields $2(2n-3)+n=5n-6$ Hermitian operators for $G$. For $n=4$ these are
\begin{tiny}
\begin{gather*}
\begin{pmatrix}
1&0&0&0\\
0&0&0&0\\
0&0&0&0\\
0&0&0&0
\end{pmatrix},
\begin{pmatrix}
0&0&0&0\\
0&1&0&0\\
0&0&0&0\\
0&0&0&0
\end{pmatrix},
\begin{pmatrix}
0&0&0&0\\
0&0&0&0\\
0&0&1&0\\
0&0&0&0
\end{pmatrix},
\begin{pmatrix}
0&0&0&0\\
0&0&0&0\\
0&0&0&0\\
0&0&0&1
\end{pmatrix},
\begin{pmatrix}
0&1&0&0\\
1&0&0&0\\
0&0&0&0\\
0&0&0&0
\end{pmatrix},
\begin{pmatrix}
0&i&0&0\\
-i&0&0&0\\
0&0&0&0\\
0&0&0&0
\end{pmatrix},
\begin{pmatrix}
0&0&1&0\\
0&0&0&0\\
1&0&0&0\\
0&0&0&0
\end{pmatrix},
\begin{pmatrix}
0&0&i&0\\
0&0&0&0\\
-i&0&0&0\\
0&0&0&0
\end{pmatrix},\\
\begin{pmatrix}
0&0&0&1\\
0&0&1&0\\
0&1&0&0\\
1&0&0&0
\end{pmatrix},
\begin{pmatrix}
0&0&0&i\\
0&0&i&0\\
0&-i&0&0\\
-i&0&0&0
\end{pmatrix},
\begin{pmatrix}
0&0&0&0\\
0&0&0&1\\
0&0&0&0\\
0&1&0&0
\end{pmatrix},
\begin{pmatrix}
0&0&0&0\\
0&0&0&i\\
0&0&0&0\\
0&-i&0&0
\end{pmatrix},
\begin{pmatrix}
0&0&0&0\\
0&0&0&0\\
0&0&0&1\\
0&0&1&0
\end{pmatrix},
\begin{pmatrix}
0&0&0&0\\
0&0&0&0\\
0&0&0&i\\
0&0&-i&0
\end{pmatrix}.
\end{gather*}
\end{tiny}
\end{example}
Finally, let us give a proof of Theorem \ref{thm2}.
\begin{proof}
For $k\in\{2r-1,\dots,2(n-r)-1\}$, define $A_k,A^\prime_k\in M(\gamma(n,k),r,\mathbb{R})$ by setting $(A_k)_{jl}=(A^\prime_k)_{jl}=x_{l+1}^j$ for all $j\in\{0,\hdots,\gamma(n,k)-1\},l\in\{0,\hdots,r-1\}$. Observe that both $A_k$ and $A^\prime_k$ can be considered as the first $r$ columns of a $\gamma(n,k)\times \gamma(n,k)$ Vandermonde matrix and since $x_j\neq x_l$ for all $j,l\in\{1,\hdots,r\}$ with $j\neq l$ and $x_l\neq 0$ for all $l\in\{1,\hdots,r\}$ they are thus totally non-singular. Applying Theorem \ref{thm3} to $A_k,A^\prime_k$ then concludes the proof. 
\end{proof}
\subsection{Proof of Theorem \ref{thm1}}
Let us now give a proof of Theorem \ref{thm1}.
\begin{proof}
Define $Y_k,X_k\in H(n),\ k\in\{1,\dots,2n-3\},$ by 
\begin{gather*}
(X_k)_{jl}:=\delta_{j+l,k}\cos\left(\frac{j-l}{2n}\pi\right),\\
(Y_k)_{jl}:=i\delta_{j+l,k}\sin\left(\frac{j-l}{2n}\pi\right),\\
j,l\in\{0,\dots,n-1\}.
\end{gather*}
Next observe two things:
\begin{enumerate}
\item[1.] The matrices $\{X_1,Y_1,\dots,X_{2n-3},Y_{2n-3}\}\subseteq H(n)$ are linearly independent by equations \eqref{eq2} and \eqref{eq3}.
\item[2.] Since $0<\frac{j-l}{2n}\pi<\frac{\pi}{2}$ for $j,l\in\{0,\dots,n-1\}$, $j>l$, we find $(X_k)_{jl}\neq 0$ and $(Y_k)_{jl}\neq 0$ for $j+l=k$, $j>l$.
\end{enumerate}
Let $u_k,w_k\in\mathbb{R}^{\gamma(n,k)},\ k\in\{1,\dots,2n-3\}$, be such that $\iota_k(u_k)=X_k-\delta_{k/2,\lceil k/2\rceil}e_{\lceil k/2\rceil}e_{\lceil k/2\rceil}^*$, $\iota_k(iw_k)=Y_k$ and note that both $u_k$ and $w_k$ have no vanishing entry. Thus, by Corollary \ref{cor3}, $\tilde{G}:=(e_0e_0^*,\dots,e_{n-1}e_{n-1}^*)\cup (X_1,Y_1,\dots,X_{2n-3},Y_{2n-3})$ is $1$-complete.

Let $G:=(e_0e_0^*,\dots,e_{n-1}e_{n-1}^*,v_1v_1^*,\overline{v}_1\overline{v}_1^*,\dots,v_{2n-3}v_{2n-3}^*,\overline{v}_{2n-3}\overline{v}_{2n-3}^*)$. To conclude the proof, we show that $\text{Span}(G)=\text{Span}(\tilde{G})$. First note that for $k\in\{1,\dots,2n-3\}$
\begin{align*}
v_kv_k^*=\sum_{j=1}^{2n-3}x_k^j (X_j+Y_j)+e_0e_0^*+x_k^{2n-2}e_{n-1}e_{n-1}^*,\\
\overline{v}_k\overline{v}_k^*=\sum_{j=1}^{2n-3}x_k^j (X_j-Y_j)+e_0e_0^*+x_k^{2n-2}e_{n-1}e_{n-1}^*
\end{align*} 
and thus $\text{Span}(G)\subseteq\text{Span}(\tilde{G})$. In order to show that $\text{Span}(\tilde{G})\subseteq\text{Span}(G)$, consider the matrix 
\begin{align*}
T:=\begin{pmatrix}
x_1&x_1^2&x_1^3&\dots&x_1^{2n-3}\\
x_2&x_2^2&x_2^3&\dots&x_2^{2n-3}\\
x_3&x_3^2&x_3^3&\dots&x_3^{2n-3}\\
\vdots&\vdots&\vdots& &\vdots\\
x_{2n-3}&x_{2n-3}^2&x_{2n-3}^3&\dots&x_{2n-3}^{2n-3}.
\end{pmatrix}
\end{align*}
The matrix $T$ is a Vandermonde matrix and thus invertible if $x_i\neq x_j$ for $i\neq j$. Hence we find \footnote{Different from the rest of the present paper, the indices we use to label $T$ begin with $1$, not with $G$.}
\begin{align*}
X_k&=\frac12\sum_{j=1}^{2n-3} (T^{-1})_{kj}(v_jv_j^*+\overline{v}_j\overline{v}_j^*-2e_0e_0^*-2x_k^{2n-2}e_{n-1}e_{n-1}^*),\\
Y_k&=\frac12\sum_{j=1}^{2n-3} (T^{-1})_{kj}(v_jv_j^*-\overline{v}_j\overline{v}_j^*)
\end{align*}
and this shows that $\text{Span}(\tilde{G})\subseteq\text{Span}(G)$.
\end{proof}

\begin{remark}
Note that there are many possible choices for the phases of the $v_i$. The only constraint is that $j\varphi\neq \frac{k\pi}{2}$ for all $j\in\{1,\dots,n-1\},\ k\in\mathbb{Z}$.
\end{remark}

\subsection{Proof of Theorem \ref{thmstab} and Proposition \ref{propstab}}
For $X_r\in\mathcal{S}^n_r$, $E\in H(n)$, $\epsilon\geq 0$ and a measurement $M$ define the set
\begin{gather}\label{spec}
 F_\epsilon(X_r,E,M):=\{Y\in \mathcal{S}^n:\|M(Y)-b\|_2\leq \epsilon\},
\end{gather}
where $b=M(X_r+E)$.
\begin{lemma}(Stability.)\label{lemstab}
Let $M$ be an $r$-complete measurement and let $\epsilon>0$. Then, there exists a constant $C_M>0$ independent of $\epsilon$ such that for all $X_r\in\mathcal{S}^n_r$,  and $E\in H(n)$ with $\|M(E)\|_2\leq\epsilon$ we have
\begin{align*}
Y\in F_\epsilon(X_r,E,M)\Rightarrow \|Y-X_r\|_2\leq C_M \epsilon.
\end{align*}
\end{lemma}

\begin{proof}
Denote by $\pi:H(n)\to \text{Range}(M^*)$ the orthogonal projection on the subspace $\text{Range}(M^*)\subseteq H(n)$ and by $\pi^\bot:H(n)\to \text{Ker}(M)$ the orthogonal projection on the subspace $\text{Ker}(M)\subseteq H(n)$. Furthermore, let $Y^\prime:=\pi^\bot(X_r)+\pi(Y)$ and let $\sigma_{min}$ be the smallest singular value of $M$ \footnote{We assume $M$ to have full rank.}. Then, we find
\begin{align}
\|X_r-Y^\prime\|&=\|\pi(X_r-Y)\|_2\leq \frac1{\sigma_{min}}\|M(Y-X_r)\|_2\nonumber\\
&\leq \frac1{\sigma_{min}}(\|M(X_r)-b\|_2+\|M(Y)-b\|_2)\leq \frac1{\sigma_{min}}(\|M(E)\|_2+\epsilon)\nonumber\\
\label{eq9}&\leq\frac2{\sigma_{min}}\epsilon.
\end{align}
From the spectral variation bound for Hermitian matrices (Theorem III.2.8 of \cite{bhatia2013matrix}) we conclude that
\begin{align*}
\|\text{Eig}(X_r)-\text{Eig}(Y^\prime)\|_2&=\sqrt{\sum_{i=1}^{r}\big(\lambda_i(X_r)-\lambda_i(Y^\prime)\big)^2+\sum_{i=r+1}^{n}\lambda_i(Y^\prime)^2}\\
&\leq\frac2{\sigma_{min}}\epsilon.
\end{align*}
But this implies that $|\lambda_i(Y^\prime)|\leq \frac2{\sigma_{min}}\epsilon$ for $i\in\{r+1,\dots,n\}$. 

Next, note that
\begin{align*}
\kappa:=-\max_{Z\in\text{Ker}(M),\|Z\|_2=1}\lambda_{n-r}(Z)
\end{align*}
exists by compactness of $\{Z\in\text{Ker}(M):\|Z\|_2=1\}$ and continuity of $\lambda_{n-r}$. Furthermore,
by Proposition \ref{proprank}, every nonzero $Z\in  \text{Ker}(M)$ has at least $r+1$ negative eigenvalues and hence we conclude that $\kappa>0$.

There exists $Z\in \text{Ker}(M)$ with  $\|Z\|_2=1$ and $\alpha\geq 0$ such that $Y=Y^\prime+\alpha Z$  \footnote{Note that $\alpha Z=\pi^\bot(Y)-\pi^\bot(X_r)\in\text{Ker}(M)$.}. Since $Y\geq 0$ we conclude from Weyl's inequality (Theorem III.2.1 of \cite{bhatia2013matrix}) that
\begin{align*}
0\leq\lambda_{n}(Y^\prime+\alpha Z)\leq \lambda_{r+1}(Y^\prime)+\alpha\lambda_{n-r}(Z)\leq \frac{2}{\sigma_{min}}\epsilon-\alpha\kappa.
\end{align*}
and hence we find
\begin{align}\label{eq10}
\alpha\leq \frac{2}{\kappa\sigma_{min}}\epsilon.
\end{align}
Finally, combining equations \eqref{eq9} and \eqref{eq10}, we conclude that
\begin{align*}
\|Y-X_r\|_2&= \|Y^\prime+\alpha Z-X_r\|_2\leq\|Y^\prime-X_r\|_2+\|\alpha Z\|_2\\
&\leq \left(\frac{2}{\sigma_{min}}+\frac{2}{\kappa\sigma_{min}}\right)\epsilon.
\end{align*}
Choosing $C_M=\frac{2}{\sigma_{min}}(1+\frac{1}{\kappa})$ then proves the claim.
\end{proof}
\begin{remark}
Since $\kappa$ just depends on $\text{Ker}(M)$, it is independent of the choice of basis for $\text{Range}(M^*)$. Thus, since it is always possible to choose an orthonormal basis of $\text{Range}(M^*)$, the constant $C_M$ is mainly determined by $\kappa$.
\end{remark}
The proof of Theorem \ref{thmstab} is an immediate consequence of this lemma.
\begin{remark}
Let $M$ be a measurement that is not $r$-complete. Then there exist $Z_r\in\mathcal{S}^n_r$ and $Z\in\mathcal{S}^n$ with $Z_r\neq Z$ such that $M(Z_r-Z)=0$ and we find $Z\in F_\epsilon(Z_r,E,M)$ for all $\epsilon>0$ and $E\in H(n)$ with $\|M(E)\|_2\leq\epsilon$. Thus, if $\id\in\text{Range}(M^*)$, the $r$-complete property is necessary to enable the recovery of every $X_r\in\mathcal{S}^n_r$ via the optimization problem \eqref{sdp2}.
\end{remark}
Finally let us give the proof of Proposition \ref{propstab}.
\begin{proof}
From Theorem \ref{thmstab} we obtain the bound $\|Y-xx^*\|_2\leq C_M\epsilon$ and the proof of Lemma \ref{stab} yields the bound $\sqrt{\sum_{i=2}^n\lambda_i(Y)^2}\le C_M\epsilon$. From this we find
\begin{align*}
\|xx^*-\hat{x}\hat{x}^*\|_2\leq \|Y-xx^*\|_2+\|Y-\hat{x}\hat{x}^*\|_2\leq 2C_M\epsilon.
\end{align*}

Finally let $\varphi\in[0,2\pi)$ be such that $\langle x,e^{i\varphi}\hat{x}\rangle$ is positive. Then,
\begin{align*}
\|x-e^{i\varphi}\hat{x}\|^2_2\|x\|_2^2&=\left(\|x\|_2^2+\|\hat{x}\|_2^2-2\text{Re}\left(\langle x,e^{i\varphi}\hat{x}\rangle\right)\right)\|x\|_2^2\\
&=\left(\|x\|_2^2+\|\hat{x}\|_2^2-2|\langle x,\hat{x}\rangle|\right)\|x\|_2^2\\
&\leq\left(\|x\|_2^2+\|\hat{x}\|_2^2-2|\langle x,\hat{x}\rangle|\right)\left(\|x\|_2^2+\|\hat{x}\|_2^2+2|\langle x,\hat{x}\rangle|\right)\\
&= \left(\|x\|_2^2+\|\hat{x}\|_2^2\right)^2-4|\langle x,\hat{x}\rangle|^2\\
&= \|x\|_2^4+\|\hat{x}\|_2^4-2|\langle x,\hat{x}\rangle|^2+2\|x\|_2^2\|\hat{x}\|_2^2-2|\langle x,\hat{x}\rangle|^2\\
&\leq 2\left(\|x\|_2^4+\|\hat{x}\|_2^4-2|\langle x,\hat{x}\rangle|^2\right)\\
&= 2\|xx^*-\hat{x}\hat{x}^*\|_2^2\\
&\leq 2(2C_M\epsilon)^2.
\end{align*}
\end{proof}
\bibliographystyle{unsrt}
\bibliography{bibliography}

\end{document}